\documentclass[a4paper,reqno,11pt]{amsart}
\usepackage{amsmath,amssymb,tabularx,setspace,color}
\usepackage[utf8]{inputenc}
\usepackage{float}
\usepackage{hyperref}
\usepackage{graphicx}

\usepackage{multicol,multirow}
\newtheorem{theorem}{Theorem}
\newtheorem{corollary}{Corollary}
\makeatletter
\renewcommand\section{\@startsection {section}{1}{\z@}%
                                   {-3.5ex \@plus -1ex \@minus -.2ex}%
                                   {2.3ex \@plus.2ex}%
                                   {\normalfont\large\bfseries}}
\makeatother
\newtheorem{Definition}{Definition}

\begin{document}

\doublespace
\title[]{A New measure of income inequality}
\author[]
{   S\lowercase{udheesh} K K\lowercase{attumannil}$^{a}$ \lowercase{and} S\lowercase{aparya} S\lowercase{uresh}$^{b, \dag}$\\
$^{a}$I\lowercase{ndian} S\lowercase{tatistical} I\lowercase{nstitute},
  C\lowercase{hennai}, I\lowercase{ndia.}\\
  $^{b}$I\lowercase{ndian}  I\lowercase{nstitute of} M\lowercase{anagement},
  K\lowercase{ozhikode}, I\lowercase{ndia.}
}
\thanks {$^{\dag}${Corresponding author E-mail: \tt saparya@iimk.ac.in}.}
\begin{abstract} A new measure of income inequality that captures the heavy tail behavior of the income distribution is proposed. We discuss two different approaches to find the estimators of the proposed measure. We show that these estimators are consistent and have an asymptotically normal distribution. We also obtain a jackknife empirical likelihood (JEL) confidence interval of the income inequality measure.    A Monte Carlo simulation study is conducted to evaluate the finite sample properties of the estimators and JEL-based confidence inerval. Finally, we use our measure to study the income inequality of three states in India.

  \noindent {\sc Keyword;} Inequality measure; Gini mean difference; Gini index; U-statistics.
\end{abstract}
\maketitle

\section{Introduction}
A large number of indices of economic inequality, compatible with various axioms of fairness, have been proposed in the literature. Most of these measures are generalizations of the Gini mean difference, placing smaller or greater weights on various portions of the income distribution. Many researchers such as Yitzhaki and Schechtman (2005, 2013), Davidson (2009), Langel and Tille (2013), Carcea and Serfling (2015), and  Sudheesh et al. (2021, 2022)   have extensively studied the Gini index and related inequality measure. Langel and Tille (2013) reviewed the literature on the Gini index and showed that many studies had repeated similar findings without citing previous research. Further, Yitzhaki and Schechtman (2013) elaborate on using Gini methodology in statistical inference and related topics. Based on Gini auto correlation, Carcea and Serfling (2015) provided a theoretical foundation for analyzing time series with heavy tail innovations.  Sudheesh et al.  (2022) proposed non parametric estimators of Gini covariance and its variants.  Sreelakshmi et al. (2021) discussed the empirical likelihood inference of the extended Gini index. In this paper, motivated by the work of Yitzhaki and Schechtman (2013) and Carcea and Serfling (2015), we propose a new measure of income inequality that captures the heavy tail behavior of the income distribution.

As mentioned, most of the income inequality measures are generalizations of the Gini mean difference/Gini index, placing smaller or greater weights on various portions of the income distribution. Thus, we start by defining the Gini index. Let $X$ be a non-negative random variable having a distribution function $F$. Assume $\mu=E(X)<\infty$.  The Gini mean difference (GMD)  is defined as
\begin{eqnarray*}
    GMD= E|X_1-X_2|,
\end{eqnarray*}where $X_1$ and $X_2$ are two random variables having the same distribution function $F.$
Then Gini index is defined as
\begin{eqnarray*}
    G= \frac{E|X_1-X_2|}{2\mu}.
\end{eqnarray*}
The extended Gini index  of order $v$ is defined as
\begin{equation*}\label{eq2.1}
  EG_v(X)=E\big(X-\min(X_1, X_2,...,X_v)\big).
\end{equation*}
where $X_1,X_2,\cdots,X_v$ are $v$ independent random variable having distribution function $F$.
A dual concept of $EG_v(X)$ is given by
\begin{equation*}\label{eq2.1}
  EG_v(-X)=E\big(\max(X_1, X_2,...,X_v)-X\big).
\end{equation*}

In financial and insurance context, the value $EG_v(X)$ is called the risk-premium and
$EG_v(-X)$ the gain-premium of $X$ of order $v$, respectively. More details of $EG_v(X)$ and $EG_v(-X)$ can be found in Cardin et al. (2013).

Using these two quantities we can define
\begin{equation*}
 \text{ The starting minimum bid}=E(X)-EG_v(X)
\end{equation*} and
\begin{equation*}
 \text{The BIN price}=EG_v(X)-E(X).
\end{equation*}
The difference between `the BIN price' and `the starting minimum bid' is called, the width of the price spread of $X$, which is an important measure in financial auctioning. It may be noted that this difference effectively captures the spread in the data and hence can play a role in modeling inequality in a given income dataset. Accordingly, we define a new income inequality measure in the following section.

The rest of the paper is organised as follows. In Section 2, we propose a new measure of income inequality and study its property. In Section 3, we discuss two methods for finding the estimators of the proposed income inequality measure; one based on U-statistics and another based on the empirical distribution function. We also study the asymptotic properties of the proposed estimators. In Section 4, we develop a jackknife empirical likelihood (JEL) based confidence interval for the proposed measure. In Section 5, we conduct a Monte Carlo simulation study to evaluate the finite sample performance of the estimators of the proposed income inequality measure. We also study the finite sample behavior of the JEL-based confidence interval. In Section 6, we illustrate the application of the proposed measure using the household income data of three states in India. Some concluding remarks are given in Section 7.

\section{Proposed Measure}

In this section, we propose a new income inequality measure that captures the heavy-tailed behavior of the data.
Let $X$ be a non-negative random variable having absolutely continuous distribution function  $F(x)$.  For positive integer $v> 1$, let $X_1, X_2,...X_v$ be independent random variables having the same cumulative distribution  $F$. Based on the expected difference between $\max(X_1, X_2,...,X_v)$ and $\min(X_1, X_2,...,X_v)$,  we define a new measure as a generalization of Gini index.

\begin{Definition}
Let $X$ be a non-negative random variable having absolutely continuous distribution function  $F(x)$.  Let $X_1, X_2,...X_v$ be the iid copies from $F$. Assume $v> 1$ is a positive integer. We define a generalized  inequality measure(GIM) of order $v$ given as
\begin{equation}\label{eq2.1}
  GIM(v)=\frac{E\big(\max(X_1, X_2,...,X_v)-\min(X_1, X_2,...,X_v)\big)}{E\big(\max(X_1, X_2,...,X_v)+\min(X_1, X_2,...,X_v)\big)}.
\end{equation}
\end{Definition}


\noindent
The numerator of the $GIM(v)$ is clearly the width of the  spread of $X$, and the denominator makes the proposed measure  to have a value in the interval $[0,1]$. 
Note that the above definition captures the tail behavior of the probability distribution. The term $v$ refers to a parameter that adjusts the sensitivity of the index to different parts of the income distribution. Thus, GIM is a generalization of the standard Gini index, designed to allow more flexibility in assessing inequality by weighting different parts of the income distribution differently. As $v$ increases, the index places more weight on changes in the lower and upper part of the income distribution, making it more sensitive to inequality among the poorer segments of the population as well as the upper part of the distribution. Moreover, when $v=2$,  $ GIM(v)$ reduces to the Gini index.\\
 Next, we study the properties of $ GIM(v)$. The proofs of the trivial cases are not presented explicitly.

\noindent {\bf Property 1:} $0\leq GIM(v) \leq 1$. \\
{\bf Proof:} Result follows by noting that $$\max(X_1, X_2,...,X_v)\geq \min(X_1, X_2,...,X_v)\geq 0.$$

\noindent {\bf Property 2:} If all the individual income in the population are equal, then $GIM(v)=0$.

\noindent {\bf Property 3:} For $v=2$, $GIM(v)$ reduces to the Gini index.
\begin{proof}

 Let $X_1$ and $X_2$ be independent random variables having  distribution function $F$. Recall, the definition of the Gini index,
\begin{equation*}\label{eq2.2}
  G=\frac{E|X_1-X_2|}{2\mu}.
\end{equation*}

\noindent From Xu (2007) and Sudheesh et al. (2021) we have the following alternative expressions for the Gini index
\begin{equation}\label{eq2.3}
  G=\frac{E(\max(X_1,X_2)-X_1)}{\mu}.
\end{equation}
\begin{equation}\label{eq2.4}
  G=\frac{E(X_1-\min(X_1,X_2))}{\mu}.
\end{equation}
From (\ref{eq2.3}) and (\ref{eq2.4}), we obtain
\begin{equation}\label{eq2.5}
  G=\frac{E(\max(X_1,X_2)-\min(X_1,X_2))}{2\mu}.
\end{equation}
Note that the distribution function of the random variable $Z_1=\max(X_1,X_2)$ is given by
\begin{equation*}\label{eq2.7}
  F_{Z_1}(x)=F^2(x).
\end{equation*}
The survival function of the random variable $Z_2=\min(X_1,X_2)$ is given by
\begin{equation*}\label{eq2.7}
  \bar{F}_{Z_2}(x)=\bar{F}^2(x),
\end{equation*}
where $\bar{F}(x)=1-F(x)$ is the survival function of $X$ at $x$. For a non-negative random variable $X$, we have $E(X)=\int_{0}^{\infty} \bar {F}(x)dx$.  Hence
 \begin{eqnarray}\label{eq2.8}
    E\big(\max(X_1, X_2)+\min(X_1, X_2)\big)&=& \int_0^\infty (1-F^2(x))dx+\int_0^\infty \bar{F}^2(x)dx\nonumber\\
    &=&\int_0^\infty 2\bar{F}(x)dx=2\mu.
    \end{eqnarray}
Substituting (\ref{eq2.8}) in (\ref{eq2.5}) we have the expression (\ref{eq2.1}) for $v=2$. Hence the result.

\end{proof}

\section{Estimation and Asymptotic properties}
We discuss two different methods for finding the estimators of $GIM(v)$; one based on U-statistics and another based on the empirical distribution function. In the first case, since $GIM(v)$ is defined as the ratio of two quantities involving the expectation of a function of random variables, finding a U-statistics-based estimator is quite straightforward and studying the asymptotic properties of the estimator is simple. Xu (2007) gives a  detailed discussion of the estimation of different inequality measures based on U-statistics.

In the second method, the GIM is expressed as an integral of a quantity involving the underlying distribution function, which is then estimated by replacing the distribution function with the empirical distribution function. Studying the asymptotic properties of these estimators is not simple and requires several algebraic manipulations. Since the empirical distribution function is a consistent and sufficient estimator of the cumulative distribution function, this method has its relevance.

\subsection{Estimation based on U-statistics}

First, we find an estimator of $GIM (v)$ based on U-statistics. We estimate the numerator and denominator of $GIM(v)$ separately.

The numerator of $GIM(v)$ can be expressed as $ N=E(h_1(X_1,...,X_v))$ where $h_1(X_1,...,X_v)=max(X_1,...,X_v)-min(X_1,...,X_v)$ is a symmetric kernel.
Hence the estimator of $N$ based on U -statistics is given by
\begin{equation}\label{2.5}
U_1=\frac{1}{\binom{n}{v}}\sum_{P_{n,v}}h_1(X_1,...,X_v),
\end{equation}
where $P_{n,v}$ is the set of all permutation of $v$ from the $\{1, \cdots, n\}$. By definition, $U_1$  is an unbiased estimator of $N$.

Similarly, an unbiased estimator of the denominator $D$ of $GIM(v)$ is given by
\begin{equation}\label{2.5}
U_2=\frac{1}{\binom{n}{v}}\sum_{P_{n,v}}h_2(X_1,...,X_v)
\end{equation}
\noindent
where $h_2(X_1,...,X_v)=max(X_1,...,X_v)+min(X_1,...,X_v)$.
Hence, the estimator of  $GIM(v)$ is given by
\begin{equation}\label{2.5}
\widehat{GIM}(v)=\frac{U_1}{U_2}.
\end{equation}
Although $U_1$ and $U_2$ are unbiased estimators of the numerator and the denominator of $GIM(v)$, their ratio $\widehat{GIM}(v)$  is not an unbiased estimator of the $GIM(v)$. The finite sample performance of  $\widehat{GIM}(v)$ is evaluated using a Monte Carlo simulation study and the results are presented in Section 5.

As the proposed estimator is based on U-statistics, we use the asymptotic theory of U-statistics to discuss the limiting behavior of $\widehat{GIM}(v)$. Since U-statistics are consistent estimators, we have  the following result.
\begin{theorem} \label{the3.1}The  $U_1$  and $U_2$ are consistent estimators of $N$ and $D$, respectively.
\end{theorem}
\begin{corollary} The $\widehat{GIM}(v)$ is a consistent estimator of $GIM(v)$.
\end{corollary}
\begin{proof}
As we can write
\begin{equation*}
  \frac{\widehat{GIM}(v)}{GIM(v)}=\frac{U_1}{U_2}.\frac{D}{N},
\end{equation*}the proof is an immediate consequence of Theorem \ref{the3.1}.
\end{proof}

\noindent Next we find the asymptotic distribution of $\widehat{GIM}(v)$. For this purpose, we find the asymptotic distribution of $U_1$
\begin{theorem} As $ n \rightarrow \infty $, $\sqrt{n}(U_1-N)$, convergence in distribution to a Gaussian random variable with mean zero and variance $v^2\sigma_{1}^{2}$, where $\sigma_1^2$ is  given by
\begin{eqnarray}
    \label{ustatsigma}
\sigma_{1}^{2}&=&Var\Big(X(F^{v-1}(X)-\bar F^{v-1}(X))\nonumber\\&&+(v-1)\int_{X}^{\infty} y F^{v-1}(y)dF(y)-(v-1)\int_{0}^{X} y \bar F^{v-1}(y)dF(y)\Big).
\end{eqnarray}
\end{theorem}
\begin{proof}
    By CLT for the U-statistics, we have the asymptotic normality of $\sqrt{n}(U_1-N)$. The asymptotic variance (Lee, 2019)  is equal to $v^2\sigma_1^2$ where $\sigma_1^2$ is given by
    \begin{equation}
        \sigma_1^2=Var\left(E(h_1(X_1,\ldots,X_v)|X_1)\right).
    \end{equation}
Denote $Z_1=\max(X_2,\ldots,X_v)$  and $Z_2=\min(X_2,\ldots,X_v))$.  Consider
\begin{eqnarray*}
 &&\hskip-0.6inE(h_1(X_1,\ldots,X_v)|X_1=x)\\&=&E((\max(X_1,\ldots,X_v)-\min(X_1,\ldots,X_v))|X_1=x)\\
&=&E(xI(x>Z_1)+Z_1I(Z_1\geq x))-E(xI(x<Z_2)-Z_1I(Z_2\geq x))\\
&=&x(F^{v-1}(x)-\bar F^{v-1}(x))+E(Z_1I(Z_1\geq x))-E(Z_2I(Z_2\leq x).
\end{eqnarray*}
Hence, we have the variance expression specified in equation \eqref{ustatsigma}.

\end{proof}

\begin{theorem} As $ n \rightarrow \infty $, $\sqrt{n}(U_2-D)$, convergence in distribution to a Gaussian random variable with mean zero and variance $v^2\sigma_{2}^{2}$, where $\sigma_2^2$ is the asymptotic variance given by
\begin{eqnarray}
    \label{ustatsigma}
\sigma_{2}^{2}&=&Var\Big(X(F^{v-1}(X)+\bar F^{v-1}(X))\nonumber\\&&+(v-1)\int_{X}^{\infty} y F^{v-1}(y)dF(y)+(v-1)\int_{0}^{X} y \bar F^{v-1}(y)dF(y)\Big).
\end{eqnarray}
\end{theorem}

\noindent Since $U_2$ is a consistent estimator of $D$, using Slutky's theorem we have the following results.

\begin{corollary} The  distribution of $\sqrt{n}(\widehat{GIM}(v)-GIM(v))$, as $ n \rightarrow \infty $, is Gaussian with mean zero and variance $v^2\sigma^{2}$, where $\sigma^2=\sigma_1^2/D^2$.
\end{corollary}

Using Corollary 2, we obtain a normal-based confidence interval for $GIM(v)$. Let $\widehat{\sigma}$ be a consistent estimator of $\sigma$ and  $Z_{\alpha}$ denote the upper $\alpha$-th percentile point of a standard normal distribution. A $100(1-\alpha)\%$ confidence interval for  $GIM(v)$ is given by
  \begin{equation*}
   \left(\widehat{GIM}(v)-Z_{\alpha/2}\frac{\widehat{\sigma}}{\sqrt{n}},\,\,\widehat{GIM}(v)+Z_{\alpha/2}\frac{\widehat{\sigma}}{\sqrt{n}}\right).\end{equation*}

\subsection{Estimation based on empirical distribution function}

Let $(X,Y)$ be a bivariate random vector with joint distribution function $F_{XY}$. Also let $F_X$ and $F_Y$ be the respective marginal distribution functions. We assume that the first moment of these random variables is finite.
Suppose $(X_1, Y_1)$, ...,$(X_n, Y_n)$ are independent and identically distributed as the bivariate random vector $(X, Y)$. Let the $Y$ variate paired with the $i$-th ordered $X$
variate $X_{i:n}$ be denoted by $Y_{[i:n]}$ is known as the concomitant of $i$-th order statistics.
Under the above formulation, consider the statistics of the form
\begin{equation}\label{2.7}
T(F_n)=\int_{0}^{\infty}\int_{0}^{\infty}J(F_n(x))h(x,y)dF_n=\frac{1}{n}\sum_{i=1}^{n}J(\frac{i}{n})h(X_{i:n},Y_{[i:n]}),
\end{equation}
where $J$ is a bounded smooth function, $h(x, y)$ is a real valued function of $(x, y)$
and $F_ n$ is the empirical distribution function of $F$ given by
\begin{equation*}
    F_n(x)=\frac{1}{n}\sum_{i=1}^{n}I(X_i<x),
\end{equation*}where $I$ denote the indicator function.
It can be easily verified that, $T(F_n)$ is a plug-in estimator of the integral of the form
\begin{equation}\label{2.71}
T(F)=\int_{0}^{\infty}\int_{0}^{\infty}J(F_X(x))h(x,y)dF_{XY}(x,y).
\end{equation}
Some of the properties of the estimator $ T(F_n)$ are first discussed by Yang (1981) in the context of non-parametric estimation of a regression function. The asymptotic properties of the estimators $T(F_n)$ have been studied by Yang (1981)  and Sandstrom (1987).

In fact the form of the estimator (\ref{2.7}) gives a unique way to find the estimators of the quantities of interest. Accordingly, for finding the estimator of $GIM(v)$ our task is reduced to rewriting the expression in equation (\ref{eq2.1}) to the form (\ref{2.71}).

Using the density functions of $\max(X_1,\ldots,X_v)$ and $\min(X_1,\ldots,X_v)$, we can rewrite the numerator of equation (\ref{eq2.1}) as
\begin{eqnarray}\label{2.8}
N &=& v\int_{0}^{\infty}x(F^{v-1}_X(x)-\bar{F}^{v-1}_X(x))dF_X(x).
\end{eqnarray}
By taking $J=F^{v-1}_X(x))-\bar{F}^{v-1}_X(x)$ and $h(x,y)=vx$, the equation (\ref{2.8}) coincides with (\ref{2.71}). An estimator of $N$ is given by
\begin{equation}\label{2.9}
\widehat{N}=\frac{v}{n^v}\sum_{i=1}^{n}(i^{v-1}-(n-i)^{v-1})X_{i:n}.
\end{equation}
Similarly, we can estimate the denominator $D$ of $GIM(v)$ is given by
\begin{equation}\label{2.9}
\widehat{D}=\frac{v}{n^v}\sum_{i=1}^{n}(i^{v-1}+(n-i)^{v-1})X_{i:n}.
\end{equation}
Hence, the estimator of $GIM(v)$ is given by
\begin{equation}\label{2.9}
\widetilde{GIM}(v)=\frac{\sum_{i=1}^{n}(i^{v-1}-(n-i)^{v-1})X_{i:n}}{\sum_{i=1}^{n}(i^{v-1}+(n-i)^{v-1})X_{i:n}}.
\end{equation}

\noindent Next, we find the asymptotic distribution of $\widetilde{GIM}(v)$.

 Under quite mild conditions Yang
(1981) established the asymptotic normality of $\sqrt{n}(T(F_n) -E(T(F_n)))$ using Hajek's projection lemma. Using a stochastic Gateaux differential, Sandstrom (1987) proved the asymptotic normality of $\sqrt{n}(T(F_n) -T(F))$.

Next, we state a general result due to Sandstrom (1987) and apply same to obtain the asymptotic distribution of the estimators derived above.

\noindent  Let
\begin{equation}\label{2.15}
\alpha_h(x)=E(h(X,Y)|X=x)
\end{equation}and
\begin{equation}\label{2.16}
  \tau^2_h(x)=V(h(X,Y)|X=x).
\end{equation}
Also let
\begin{equation}\label{2.17}
  \sigma^2=\sigma_{11}^2+\sigma_{22}^2,
\end{equation}
 where
\begin{eqnarray}\label{2.18}
\sigma_{11}^2&=&\int_{0}^{\infty}\int_{0}^{\infty}\big[min\{F_X(x),F_X(z)\}-F_X(X)F_X(z)\big]\nonumber \\&\quad&\quad\quad\quad\quad\quad\quad\quad\quad J(F_X(x))J(F_X(z))d\alpha_h(x)d\alpha_h(z)
 \end{eqnarray}
 and
  \begin{equation}\label{2.19}
  \sigma_{22}^2=\int_{0}^{\infty}J^2(F_X(x))\tau_h^2(x)dF_X.
 \end{equation}


\begin{theorem}\label{th4.1}Assume $\alpha_h(x)$ is right continuous and that $J$ is bounded in $[0,1]$ and differentiable. Also assume that $\alpha_h(x)$ and $\tau^2_h(x)$ are finite.
Suppose, $\sigma^{2}$ is as defined in (\ref{2.17}). Then as $n\rightarrow \infty$,  $\sqrt{n}((T(F_n)- T(F))/\sigma$ converges in distribution to a standard normal random variable.
\end{theorem}

We shall now use this theorem to derive the asymptotic distribution of the estimators defined above.

\begin{corollary} As $n\rightarrow \infty$,
the distribution of $\sqrt{n}(\widehat{N}- N)/\sigma_3$ converges to standard normal distribution, where $\sigma_3$ is given by
\begin{eqnarray*}
\sigma_3^2&=&v^2\int_{0}^{\infty}\int_{0}^{\infty}\big[min\{F_X(x),F_X(z)\}-F_X(x)F_X(z)\big]\\&\quad&\quad
(F^{v-1}_X(x))-\bar{F}^{v-1}_X(x))(F^{v-1}_X(z))-\bar{F}^{v-1}_X(z))dxdz.
\end{eqnarray*}
\end{corollary}
\begin{proof}
    The asymptotic normality follows from Theorem \ref{th4.1}. Note that $\alpha_h(x)=vx$, $\tau^2_h(x)=0$ and $J(x)=F^{v-1}_X(x))-\bar{F}^{v-1}_X(x)$, hence we have the variance expression given as in the Corollary.

\end{proof}

Similarly, we have the following result for $\widehat{D}$.
\begin{corollary} As $n\rightarrow \infty$,
the distribution of $\sqrt{n}(\widehat{D}- D)/\sigma_4$ converges to standard normal distribution, where $\sigma_4$ is given by
\begin{eqnarray*}
\sigma_4^2&=&v^2\int_{0}^{\infty}\int_{0}^{\infty}\big[min\{F_X(x),F_X(z)\}+F_X(x)F_X(z)\big]\\&\quad&\quad
(F^{v-1}_X(x))+\bar{F}^{v-1}_X(x))(F^{v-1}_X(z))+\bar{F}^{v-1}_X(z))dxdz.
\end{eqnarray*}
\end{corollary}

Sudheesh et al. (2022) proved that $ T(F_n)$  is a consistent estimator of $T(F)$. Accordingly, $\widehat{D}$ is a consistent estimator of $D$. Hence using Slutky's theorem, we have the following result.

\begin{corollary} The  distribution of $\sqrt{n}(\widetilde{GIM}(v)-GIM(v))$, as $ n \rightarrow \infty $, is Gaussian with mean zero and variance $v^2\sigma_*^{2}$, where $\sigma_{*}^{2}=\sigma_{3}^2/D^2$.
\end{corollary}
Using Corollary 5, we can obtain a normal-based confidence interval for $GIM(v)$. Let $\widehat{\sigma}_*$ be a consistent estimator of $\sigma_*$ and  $Z_{\alpha}$ denote the upper $\alpha$-th percentile point of a standard normal distribution. A $(1-\alpha)$ level  confidence interval for  $GIM(v)$ is given by
  \begin{equation*}
   \left(\widetilde{GIM}(v)-Z_{\alpha/2}\frac{\widetilde{\sigma}*}{\sqrt{n}},\,\,\widetilde{GIM}(v)+Z_{\alpha/2}\frac{\widehat{\sigma}_*}{\sqrt{n}}\right).\end{equation*}

\section{Jackknife empirical likelihood based confidence interval}

In the previous section, we discussed the construction of the normal based confidence interval for $GIM(v)$. As it is difficult to find a consistent estimator of the asymptotic variance, the implementation of this method is not advisable. Hence we propose a jackknife empirical likelihood (JEL) based confidence interval for $GIM(v)$ which is distribution free.  Jackknife empirical method was proposed by Jing et al (2009). This method is a combination jackknife resampling technique and the empirical likelihood method, and it is used for constructing confidence intervals and performing hypothesis tests. JEL-based inference for income inequality measures has been studied by Wang et al. (2016), Wang and Zhao (2016), Sang et al. (2019) and Wei et al. (2022).

For implementing this method, we need to generate the jackknife pseudo values. For this purpose,  we define the estimating equation  as
 \begin{equation}\label{senjack}
   S_{n}=\frac{1}{\binom{n}{v}} \sum_{P_{n,v}} h(X_i,X_j,\ldots,X_{v};{GIM}(v))=0,
 \end{equation}
 where
 \begin{eqnarray*}
     h(X_1,X_2,..,X_{v}; GIM)&=&{GIM}(v)(\max(X_1,...,X_v)+min(X_1,...,X_v))\\&&-(\max(X_1,...,X_v)-\min(X_1,...,X_v)),
 \end{eqnarray*} and $P_{n,v}$ is the set of all permutation of $v$ from the set $\{1, \cdots, n\}$.
  Now we define the jackknife pseudo values  as
  \begin{equation*}
    \widehat{V}_{k}=nS_{n}-(n-1)S_{n-1,k};\,\,k=1,2,...,n,
  \end{equation*}
  where $S_{n-1,k}$ is calculated from (\ref{senjack}) using the sample observations\\$X_1,X_2,...,X_{k-1},X_{k+1},...,X_{n}$.
  Thus we have
  \begin{equation*}
    \frac{1}{n}\sum_{k=1}^{n}\widehat{V}_{k}=S_{n}.
  \end{equation*}
 Let ${p_k}=, k=1,2,\ldots, n$ be the probability associated with each  $\widehat{v}_k$. Define JEL for GIM measure as
  \begin{equation*}
 JEL(GIM)=\sup_{\bf p} \big(\prod_{i=1}^{n}{p_i};\,\, \sum_{i=1}^{n}{p_i}=1;\,\,\sum_{i=1}^{n}{p_i \widehat{V}_{k}}=0\big).
\end{equation*}
Note that $\prod_{i=1}^{n}{p_i}$ is maximised subject to the condition  $\sum_{i=1}^{n}{p_i}=1$ at $p_i = 1/n$.
Hence, using the Lagrange multiplier method, we obtain the jackknife empirical log-likelihood ratio as
\begin{equation*}
  \log R(GIM)=-\sum_{k=1}^{n}\log\big(1+\lambda\widehat{V}_{k}\big),
\end{equation*}
where $\lambda$ is the solution of
\begin{equation}\label{lambda2}
  \frac{1}{n}\sum_{k=1}^{n}{\frac{\widehat{V}_{k}}{1+\lambda \widehat{V}_{k}}}=0,
\end{equation}
    provided
\begin{equation*}
  \min_{{1\le k\le n}}\widehat{V}_{k}<{GIM}(v)<  \max_{1\le k\le n}\widehat{V}_{k}.
\end{equation*}
Next we obtain the asymptotic distribution of the jackknife empirical log-likelihood ratio as an analog of Wilk's theorem.
\begin{theorem}
 Denote $g(x)=E\left(\psi(X_1,X_2,\ldots,X_{v};GIM(v))|X_1=x\right)$
 and $\sigma_{g}^{2}=Var(g(X))$.
Suppose that $E\left(\psi(X_1,X_2,\ldots,X_{v};GIM(v))\right)<\infty$  and $\sigma_{g}^{2}>0$. Then as $n\rightarrow \infty$, $-2\log R(GIM)$ converges in distribution to a $\chi^2$ random variable  with one degree of freedom.
\end{theorem}
Using Theorem 5, we can construct  a $100(1-\alpha) \%$ JEL based confidence interval for $GIM$ as
\begin{equation*}
  \left\{GIM:-2\log R(GIM)\leq \chi^2_{1-\alpha}(1)\right\},
\end{equation*}where $\chi^2_{(1-\alpha)}(1)$ is the $(1-\alpha)-$th  percentile point of a $\chi^2$  distribution with one degree of freedom.

\section{Simulation Study}
In this section, we carried out a Monte Carlo Simulation study to evaluate the finite sample performance of the proposed estimators of $GIM(v)$. The simulation is done using  R package and repeated ten thousand times. We carried out the simulation study for $v=2$ and $v=3$ with three different distributional assumptions of the $X$- exponential, Pareto and lognormal. The results from the simulation study for these distributions are given in Table \ref{tab:exp},\ref{tab:pareto}, \ref{tab:lnorm}, respectively. For each case, the bias and mean square deviation (MSE) is calculated to evaluate the performance of the proposed estimators.

\begin{table}[h]
\caption{Bias and MSE: Exponential distribution}
    \centering
    \begin{tabular} {| c | c | c | c | c | c |}
    \hline
  &  & \multicolumn{2} {|c|} {$\widehat{GIM}(v)$} & \multicolumn{2} {|c|} {$\widetilde{GIM}(v)$} \\
 \hline
\multirow{7}{*}{ $v=2$} & $n$ & Bias  & MSE & Bias & MSE\\
\hline

   &20 & 0.000 &0.004 &0.025 & 0.004 \\
    &40 & 0.000 & 0.002&0.014 &0.002\\
    &60 & 0.000 & 0.001& 0.007&0.001\\
    &80 & 0.000 & 0.001&0.005 &0.000\\
    &100 & 0.000 & 0.000& 0.004 &0.000\\
    &200 & 0.000 & 0.000& 0.002 &0.000\\
    \hline
    \multirow{7}{*}{$v=3$} & $n$ & Bias  & MSE & Bias & MSE\\
\hline
   &20 & -0.009 & 0.005 &0.018 & 0.004 \\
    &40 & -0.003 & 0.002&0.008 &0.002\\
    &60 & -0.001 & 0.001& 0.004&0.002\\
    &80 & 0.000 & 0.001&0.003 &0.001\\
    &100 & 0.000 & 0.000& 0.004 &0.000\\
    &200 & 0.000 & 0.000& 0.001 &0.000\\
    \hline
    \end{tabular}

    \label{tab:exp}
\end{table}

\begin{table}[]
\caption{Bias and MSE: Pareto distribution}
    \centering
    \begin{tabular} {| c | c | c | c | c | c |}
    \hline
  &  & \multicolumn{2} {|c|} {$\widehat{GIM}(v)$} & \multicolumn{2} {|c|} {$\widetilde{GIM}(v)$} \\
 \hline
\multirow{7}{*}{ $v=2$} & $n$ & Bias  & MSE & Bias & MSE\\
\hline

   &20 & -0.032 &0.012 &0.008 & 0.011 \\
    &40 & -0.021 & 0.008&-0.003 &0.007\\
    &60 & -0.014 & 0.006& -0.003&0.005\\
    &80 & -0.010 & 0.005&-0.001 &0.005\\
    &100 & -0.012 & 0.004&-0.001 &0.004\\
    &200 & -0.005 & 0.002& -0.004 &0.002\\
    \hline
    \multirow{7}{*}{$v=3$} & $n$ & Bias  & MSE & Bias & MSE\\
\hline
   &20 & -0.048&0.016&0.003 & 0.012  \\
    &40 &-0.025 &0.010&-0.005 & 0.008\\
    &60 & -0.020&0.008&-0.003 & 0.007\\
    &80 &-0.010 &0.006&-0.006 & 0.006\\
    &100 & -0.010&0.005& -0.002 & 0.005\\
    &200 & -0.009&0.003&-0.002 & 0.003\\
    \hline
    \end{tabular}

    \label{tab:pareto}
\end{table}

\begin{table}
\caption{Bias and MSE: Lognormal distribution}
    \centering
    \begin{tabular} {| c | c | c | c | c | c |}
    \hline
  &  & \multicolumn{2} {|c|} {$\widehat{GIM}(v)$} & \multicolumn{2} {|c|} {$\widetilde{GIM}(v)$} \\
 \hline
\multirow{7}{*}{ $v=2$} & $n$ & Bias  & MSE & Bias & MSE\\
\hline

   &20 & 0.000 &0.002 &0.033 & 0.002 \\
    &40 & 0.000 & 0.001&0.010 &0.001\\
    &60 & 0.000 & 0.001& 0.011&0.001\\
    &80 & 0.000 & 0.000&0.007 &0.001\\
    &100 & 0.000 & 0.000&0.006&0.000\\
    &200 & 0.000 & 0.000& 0.003 &0.000\\
    \hline
    \multirow{7}{*}{$v=3$} & $n$ & Bias  & MSE & Bias & MSE\\
\hline
   &20 & 0.000&0.003&0.040 & 0.004  \\
    &40 &0.000 &0.001&-0.020 & 0.001\\
    &60 & 0.000&0.001&-0.015 & 0.001\\
    &80 &0.000 &0.001&-0.010 & 0.000\\
    &100 & 0.000&0.001& -0.008 & 0.000\\
    &200 & 0.000&0.000&-0.004 & 0.000\\
    \hline
    \end{tabular}

    \label{tab:lnorm}
\end{table}

From  Tables \ref{tab:exp}, \ref{tab:pareto}, \ref{tab:lnorm}, we can see that the bias and MSE for both estimators are steadily converging to zero as the sample size increases.

Next we study the performance of JEL based confidence interval for the proposed GIM measure. We find the coverage probability when $X$ is generated from exponential, Pareto and lognormal distributions. The results of the simulation studies are reported in Tables \ref{tab:cpexp} - \ref{tab:cplnorm}. From these tables, we observe that the coverage probabilities converges to the confidence level as the sample size increases.

\begin{table}[]
\caption{Coverage probabilities of confidence interval for exponential distribution with $\lambda=1$  }
    \centering
    \begin{tabular} {cccccccccc}
 \hline
 & $n$ & $v=2$ & $v=3$& $v=4$ \\
\hline

   &20 & 0.931& 0.942 & 0.937   \\
    &40 & 0.949& 0.947 & 0.941\\
    &60 & 0.949 & 0.950& 0.949 \\
    &80 & 0.950& 0.950& 0.950\\
    &100 & 0.950&  0.950& 0.950\\
    \hline
    \end{tabular}
    \label{tab:cpexp}
\end{table}

\begin{table}[]
\caption{Coverage probabilities of confidence interval for Pareto distribution with $\alpha=2$  }
    \centering
    \begin{tabular} {cccccccccc}
 \hline
 & $n$ & $v=2$ & $v=3$& $v=4$ \\
\hline

   &20 & 0.701& 0.696 & 0.723   \\
    &40 & 0.851& 0.771 & 0.843\\
    &60 & 0.893 & 0.934& 0.917 \\
    &80 & 0.949& 0.948& 0.946\\
    &100 & 0.950&  0.950& 0.950\\
    \hline
    \end{tabular}
    \label{tab:cppar}
\end{table}

\begin{table}[]
\caption{Coverage probabilities of confidence interval for lognormal distribution with $(\mu,\sigma)=(0,1)$  }
    \centering
    \begin{tabular} {cccccccccc}
 \hline
 & $n$ & $v=2$ & $v=3$& $v=4$ \\
\hline

   &20 & 0.882& 0.919 & 0.911   \\
    &40 & 0.931& 0.941 & 0.934\\
    &60 & 0.949 & 0.950& 0.949 \\
    &80 & 0.950& 0.950& 0.950\\
    &100 & 0.950&  0.950& 0.950\\
    \hline
    \end{tabular}
    \label{tab:cplnorm}
\end{table}

It is worth investigating what is the ideal choice of $v$.  We propose an elbow method for choosing the value of $v$. In this method, we plot the values of $GIM(v)$ for different $v$'s and we find the point where the curve bends and starts to flatten, resembling an elbow which is chosen as the optimal $v$. To understand this, we make a plot to observe the behavior of $GIM(v)$ as $v$ increases. We have generated $X$ from exponential ($\lambda =1$), Pareto ($\alpha =2$0 and lognormal $((\mu,\sigma) =(0,1))$ distributions  and calculated the GIM(v) values for different values of $v$. The result is presented in Figure \ref{fig:gim}. From these plots, we recommend to use $v = 3$ for all the cases considered here.

\begin{figure}[H]
\centering
\caption{GIM(v) for different distributions}
\vspace{0.2in}
\begin{tabular}{c c }
     \includegraphics[width=6 cm]{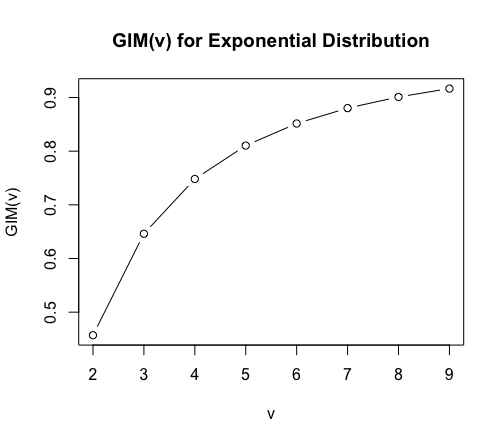}

 &\includegraphics[width=6 cm]{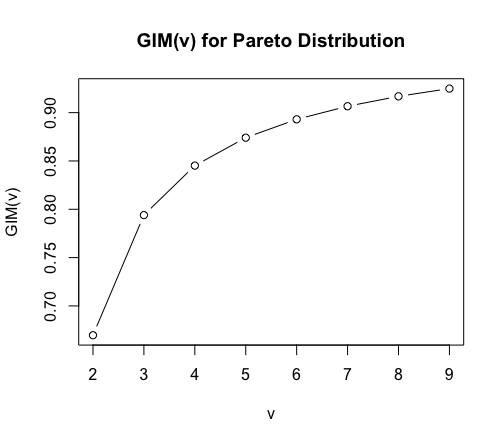}

 \\

\end{tabular}
\includegraphics[width=6 cm]{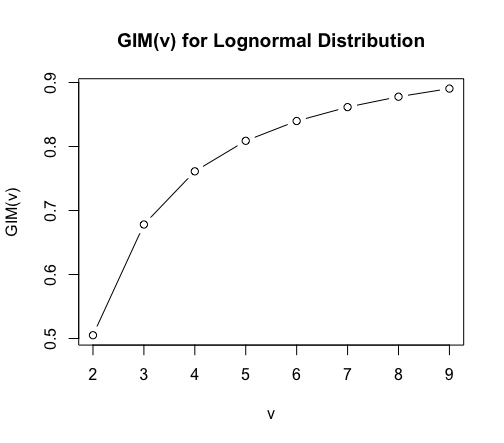}

\label{fig:gim}

\end{figure}

\section{Data Analysis}

We have selected state-by-state income data from India for the purpose of empirical illustration. We selected Kerala, Tamil Nadu, and Bihar as the three states for the study because they each have a varied amount of income disparity due to the socio-economic conditions that prevail there. The Consumer Pyramids Household Survey (CPHS) of the Centre for Monitoring Indian Economy (CMIE) is the source of the household-level income data for each state (the dataset is available at \url{https://consumerpyramidsdx.cmie.com}. It is a frequent, comprehensive survey that is conducted on a regular basis to obtain data about Indian household demographics, spending, assets, and attitudes. Every year, three waves of data are collected, each lasting four months. We used data from Wave 28, which comprises information gathered between January and April, 2023 for the analysis.

The distribution of income for each state is given in Figure \ref{fig:den}. From the distributional pattern, it can be inferred that there is an increasing pattern of income inequality with Tamil Nadu being the lowest and the Bihar being the highest. To further understand the inequality patterns in each state we have reported the descriptive statistics of the income data in Table  \ref{tab:des} and the Gini index is calculated using R-package `dineq' and is reported in Table \ref{tab:gini}. The descriptive statistics show the presence of heavy tails in samples from Kerala and Bihar compared to Tamil Nadu. Though the Gini index provides a summary of the inequality, from Figure \ref{fig:den} and Table  \ref{tab:des}, we can infer the presence of more inequality in Bihar.
 Hence, we feel that the proposed measure can effectively capture the dispersion in the income as the calculation of the $GIM(v)$  includes the extreme values in the datasets and hence captures the heavy tail behavior of the income distribution.

\begin{figure}[htbp]
\centering
\caption{Income distribution of different states}
\vspace{0.2in}
\begin{tabular}{c c }
     \includegraphics[width=6 cm]{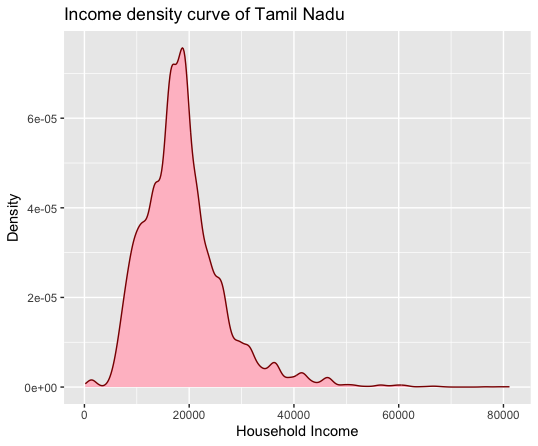}

 &\includegraphics[width=6 cm]{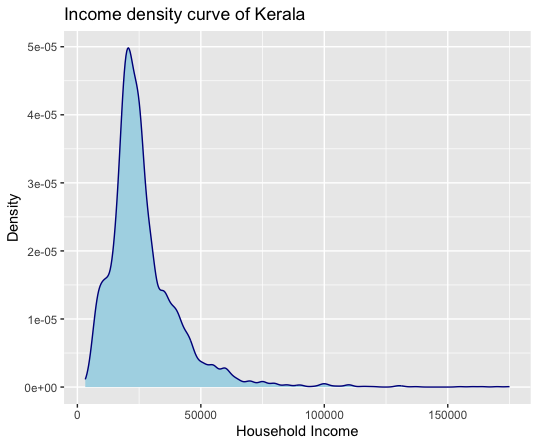}

 \\

\end{tabular}
\includegraphics[width=6 cm]{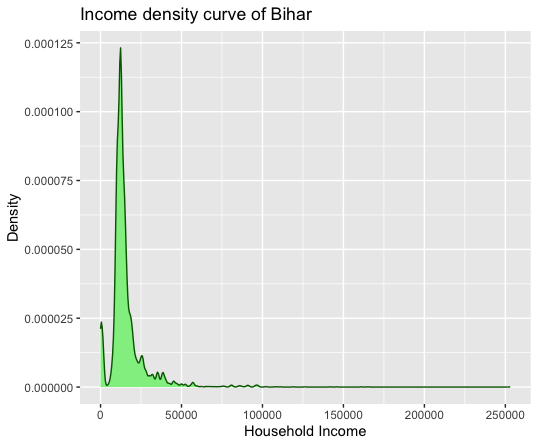}

\label{fig:den}

\end{figure}

\begin{table}[H]
    \centering
     \caption{Descriptive statistics}
    \begin{tabular}{|c|c|c|c|}
    \hline
     & Tamil Nadu   & Kerala & Bihar \\
     \hline

      $n$   & 8129 & 4310 & 7475\\
      \textit{Mean}  & 18736.31 & 26829.06 &15713.65\\

      \textit{SD} & 7792.97 & 15185.06 & 12018.63 \\
     \textit{ Min }& 165 &3200 &0\\
      \textit{Max} &81150& 175000 &253000\\
      \textit{Range} & 80985& 171800 & 253000\\
      \textit{Skewness} & 1.53 & 2.75 & 4.69 \\
     \textit{Kurtosis} & 5.19 &14.4 &43.5\\
      \hline

    \end{tabular}

    \label{tab:des}
\end{table}

\begin{table}[H]
    \centering
     \caption{The Gini index of each state in India}
    \begin{tabular}{|c|c|}
    \hline
     State    &  Gini Index\\
     \hline
    Tamil Nadu     & 0.216\\

    Kerala & 0.271\\

    Bihar& 0.310\\
    \hline
    \end{tabular}

    \label{tab:gini}
\end{table}

\begin{table}[H]
    \centering
    \caption{$GIM (v)$ values calculated for $v=2$ and $v=3$}
    \begin{tabular}{|c|c|c|}
    \hline
      State   & $GIM(2)$ & $GIM(3) $\\
     \hline
    Tamil Nadu     & 0.216&0.319\\

    Kerala & 0.271 & 0.393\\

    Bihar& 0.310 & 0.441\\
    \hline
    \end{tabular}

    \label{tab:stat_gim}
\end{table}
Table \ref{tab:stat_gim} gives the $GIM(v)$ calculated for the chosen states for $v=2 $ and $v=3$. As was previously discussed, when $v=$2, $GIM(2)$ becomes the Gini Index, and we can observe that the values generated from the real dataset also coincide.  It can be seen that like the Gini Index, the proposed measure also gives the same pattern of inequality among the chosen states. However, we can see that $GIM(3)$ values are higher than the respective $GIM(2)$. This shows that $GIM(3)$ might be capturing the dispersion much more effectively compared to $GIM(2)$.

\section{Conclusion}
There are many inequality measures available in the literature which are generalizations of the Gini mean difference or Gini index. Here we proposed a new measure 'Generalised inequality measure (GIM)' as a generalization of the Gini index.  The proposed measure captures the effect of extreme observations in the sample. We introduced two different estimators of the proposed measured and also studied the asymptotic properties of these estimators. We also developed a JEL based confidence interval for $GIM(v)$. The proposed measure was calculated for a set of real data in the study. The GIM can be modified for truncated random variables, allowing us to study the distribution of income among  affluent or poor people.







\end{document}